\documentclass[11pt,a4paper]{article}

\usepackage[dvips]{graphicx}
\usepackage{amsmath}
\usepackage{amsthm}
\usepackage{amssymb}
\usepackage{verbatim}
\usepackage{algorithmic}
\usepackage{algorithm}
\usepackage{enumerate}
\usepackage{framed}
\usepackage{marginnote}
\usepackage{color}

\usepackage{hyperref}

\newtheorem{theorem}{Theorem}

\newtheorem{lemma}[theorem]{Lemma}

\newtheorem{myclaim}{Claim}

\newcommand{\bs}{\bar S}

\newcommand{\R}{\mathbb{R}}
\newcommand{\Q}{\mathbb{Q}}

\newcommand{\mk}{\textsc{MinKnap}}
\newcommand{\maxk}{\textsc{MaxKnap}}

\newcommand{\Zz}{\mathbb{Z}}
\newcommand{\N}{\mathbb{N}}

\newcommand{\qedd}{\hfill $\blacksquare$}

\usepackage{mathtools}
\DeclarePairedDelimiter\ceil{\lceil}{\rceil}

\newcommand*\samethanks[1][\value{footnote}]{\footnotemark[#1]}

\title{\textbf{On bounded pitch inequalities for the \\ min-knapsack polytope}}

\bigskip

\author{Yuri Faenza\thanks{IEOR, Columbia University, USA. Email: yf2414@columbia.edu} \and Igor Malinovi\'c\thanks{SB/IC, EPFL, Switzerland. Email: firstname.lastname@epfl.ch} \and Monaldo Mastrolilli\thanks{IDSIA, USI-SUPSI, Switzerland. Email: monaldo@idsia.ch} \and Ola Svensson\samethanks[2]}
\date{\today}
\begin{document}
\maketitle
\begin{abstract}
In the min-knapsack problem one aims at choosing a set of objects with minimum total cost and total profit above a given threshold. In this paper, we study a class of valid inequalities for min-knapsack known as bounded pitch inequalities, which generalize the well-known unweighted cover inequalities. While separating over pitch-1 inequalities is NP-hard, we show that approximate separation over the set of pitch-1 and pitch-2 inequalities can be done in polynomial time. We also investigate integrality gaps of linear relaxations for min-knapsack when these inequalities are added. Among other results, we show that, for any fixed $t$, the $t$-th CG closure of the natural linear relaxation has the unbounded integrality gap.

\bigskip

\end{abstract}

\section{Introduction}

The min-knapsack problem (\mk)
\begin{equation}\label{eq:mink}\min c^T x \; \; \hbox{ s.t. }\; \; p^T x \geq 1, \; x \in \{0,1\}^n
\end{equation}
is the variant of the max knapsack problem (\maxk) where, given a cost vector $c$ and a profit vector $p$, we want to minimize the total cost given a lower bound on the total profit. \mk \ is known to be NP-Complete, even when $p=c$. Moreover, it is easy to see that the classical FPTAS for \maxk \ \cite{Ibarra1975,Lawler1979} can be adapted to work for \mk, thus completely settling the complexity of \mk.

However, \emph{pure} knapsack problems rarely appear in applications. Hence, one aims at developing techniques that remain valid when less structured constraints are added on top of the original knapsack one. This can be achieved by providing \emph{strong} linear relaxations for the problem: then, any additional linear constraint can be added to the formulation, providing a good starting point for any branch-and-bound procedure. The most common way to measure the strength of a linear relaxation is by measuring its \emph{integrality gap}, i.e. the maximum ratio between the optimal solutions of the linear and the integer programs (or of its inverse if the problem is in minimization form) over all the objective functions.

Surprisingly, if we aim at obtaining linear relaxations with few inequalities and bounded integrality gap, \mk\ and \maxk\ seem to be very different. Indeed, the standard linear relaxation for \maxk\ has integrality gap 2, and this can be boosted to $(1+\epsilon)$ by an extended formulation with $
n^{\bar{O}(1/ \epsilon)}$ many variables and constraints, for $\epsilon>0$ \cite{Bienstock2008}. Conversely, the standard linear relaxation for \mk\, has unbounded integrality gap, and this remains true even after $\Theta(n)$ rounds of the Lasserre hierarchy \cite{Mastrolilli2015}. No linear relaxation for \mk \, with polynomially many constraints and constant integrality gap can be obtained in the original space \cite{Carsten2015}. It is an open problem whether an extended relaxation with this properties exists. Recent results showed the existence \cite{Bazzi2017} and gave an explicit construction \cite{Fiorini2017} of a linear relaxation for \mk\, of quasi-polynomial size with integrality gap $2+\epsilon$. This is obtained by giving an approximate formulation for \emph{Knapsack Cover inequalities} (KC) (see \cite{Carr2000} and the references therein). Adding those (exponentially many) inequalities gives an integrality gap of $2$, and can be approximately separated \cite{Carr2000}. The bound on the integrality gap is tight, even in the simpler case when $p=c$. One can then look for other classes of well-behaved inequalities that can be added to further reduce the integrality gap. A prominent family is given by the so called \emph{bounded pitch inequalities} \cite{Bienstock2006} defined in Section \ref{sec:basics}. Here, we remark that the \emph{pitch} is a parameter measuring the complexity of an inequality, and the associated separation problem is NP-Hard already for pitch-$1$. The pitch-$1$ inequalities are often known in the literature as \emph{unweighted cover inequalities} (see e.g. \cite{Bazzi2017}).


In this paper, we study structural properties and separability of bounded pitch inequalities for \mk, and the strength of linear relaxations for \mk \ when they are added. Let ${\cal F}$ be the set given by pitch-$1$, pitch-$2$, and inequalities from the linear relaxation of \eqref{eq:mink}.  We first show that, for any arbitrarily small precision, we can solve in polynomial time the \emph{weak separation problem} for the set ${\cal F}$. Even better, our algorithm either certifies that the given point $x^*$ violates an inequality from ${\cal F}$, or outputs a point that satisfies all inequalities from ${\cal F}$ and whose objective function value is arbitrarily close to that of $x^*$. We define such an algorithm as a \emph{$(1+\epsilon)$-oracle} in Section \ref{sec:basics}; see Section \ref{sec:separation} for the construction. A major step of our procedure is showing that non-redundant pitch-2 inequalities have a simple structure.

It is then a natural question whether bounded pitch inequalities can help to reduce the integrality gap below $2$. We show that, when $p=c$, if we add to the linear relaxation of \eqref{eq:mink} pitch-$1$ and pitch-$2$ inequalities, the integrality gap is bounded by $3/2$; see Section \ref{sec:p-c}. However, this is false in general. 
Indeed, we also prove that KC plus bounded pitch inequalities do not improve upon the integrality gap of $2$; see Section \ref{sec:KC+bounded}. Moreover, bounded pitch alone can be much weaker than KC: we show that, for each fixed $k$, the integrality gap may be unbounded even if all pitch-$k$ inequalities are added. Using the relation between bounded pitch and Chv\'atal-Gomory (CG) closures established in \cite{Bienstock2006}, this implies that, for each fixed $t$, the integrality gap of the $t$-th CG closure can be unbounded; see Section \ref{sec:CG}. For an alternative proof that having all KC inequalities bounds the integrality gap to $2$ see Section \ref{sec:kc-gap-2}. 

\section{Basics}\label{sec:basics}

A \mk \ instance is a binary optimization problem of the form \eqref{eq:mink},
where $p, c \in \Q^n$ and we assume $0\leq p_1 \leq p_2 \leq \dots \leq p_n \leq 1$, $0 < c_i \leq 1, \ \forall i \in [n]$. We will often deal with its \emph{natural linear relaxation}
\begin{equation}\label{eq:minkLR}\min c^T x \; \; \hbox{ s.t. }\; \; p^T x \geq 1, \; x \in [0,1]^n.\end{equation}
The NP-Hardness of \mk \ immediately follows from the fact that \maxk \ is NP-Hard \cite{Karp1972}, and that a \maxk \ instance
\begin{equation}\label{eq:maxk}\max v^T x \; \; \hbox{ s.t. }\; \; w^T x \leq 1, \; x \in \{0,1\}^n.\end{equation}
can be reduced into a \mk \ instance \eqref{eq:mink} as follows: each $x \in \{0,1\}^n$ is mapped via $\pi: \R^n \rightarrow \R^n$ with $\pi(x) = \mathbf{1}-x$; $p_i = \frac{w_i}{\sum_{i=1}^n w_i - 1}$ and $c_i=v_i$ for $i \in [n]$. Note that the reduction is not approximation-preserving.

We say that an inequality $w^T x \geq \beta$ with $w\geq 0$ is \emph{dominated} by a set of inequalities ${\cal F}$ if $w'^T x \geq \beta'$ can be written as a conic combination of inequalities in ${\cal F}$ for some $\beta' \geq \beta$ and $w' \geq w$. $w^T x \geq \beta$ is \emph{undominated} if any set of valid inequalities dominating $w^T x \geq \beta$ contains a positive multiple of it.

Consider a family ${\cal F}$ of inequalities valid for \eqref{eq:mink}. We refer to \cite{GLS} for the definition of \emph{weak separation oracle}, which is not used in this paper. We say that ${\cal F}$ admits a \emph{$(1+\epsilon)$-oracle} if, for each fixed $\epsilon>0$, there exists an algorithm that takes as input a point $\bar x$ and, in time polynomial in $n$, either outputs an inequality from ${\cal F}$ that is violated by $\bar x$, or outputs a point $\bar y$, $\bar x \leq \bar y \leq (1+\epsilon)\bar x$ that satisfies all inequalities in $\cal F$. In particular, if ${\cal F}$ contains the linear relaxation of \eqref{eq:mink}, $0\leq \bar y \leq 1$. 

Let $\sum_{i \in T} w_i x_i \geq \beta $ be a valid inequality for \eqref{eq:mink}, with $w_i > 0$ for all $i \in T$. Its \emph{pitch} is the minimum $k$ such that, for each $I \subseteq T$ with $|I|=k$, we have $\sum_{i \in I} w_i \geq \beta $. Undominated pitch-1 inequalities are of the form $\sum_{i \in T} x_i \geq 1$. Note that the map from \maxk \ to \mk \ instances defined above gives a bijection between \emph{minimal cover inequalities}
$$ \sum_{i \in I} x_i \leq |I|-1$$
for \maxk \ and undominated pitch-1 inequalities for the corresponding \mk \ instance. Since, given a \maxk \ instance, it is NP-Hard to separate minimal cover inequalities \cite{Klabjan1998}, we conclude the following.

\begin{theorem}
It is NP-Hard to decide whether a given point satisfies all valid pitch-1 inequalities for a given \mk \ instance.
\end{theorem}

Given a set $S\subseteq [n]$, such that $\beta:= 1-\sum_{i \in S}p_i > 0$, the \emph{Knapsack cover inequality} associated to $S$ is given by
\begin{equation}
\sum_{i \in [n]\setminus S} \min\{p_i,\beta\} x_i \geq \beta
\end{equation}
and it is valid for \eqref{eq:mink}.

For a set $S\subseteq [n]$, we denote by $\chi^S$ its characteristic vector. An \emph{$\epsilon$-approximate solution} for a minimization integer programming problem is a solution $\bar x$ that is feasible, and whose value is at most $(1+\epsilon)$ times the value of the optimal solution. An algorithm is called a \emph{polynomial time approximation scheme (PTAS)} if for each $\epsilon > 0$ and any instance of the given problem it returns an $\epsilon$-approximate solution in time polynomial in the size of the input. If in addition the running time is polynomial in $1/ \epsilon$, then the algorithm is a \emph{fully polynomial time approximation scheme (FPTAS)}.

Given a rational polyhedron $P=\{x \in \R^n: Ax \geq b \}$ with $A\in \Zz^{m \times n}$ and $b\in \Zz^{m}$, the first \emph{Chv\'{a}tal-Gomory (CG) closure} \cite{CCZ} of $P$ is defined as follows:
$$P^{(1)}=\{x \in \R^n: \ \lceil \lambda^\top A \rceil x \geq \lceil \lambda^\top b\rceil, \ \forall \lambda \in \R^m \}.$$
Equivalently, one can consider all $\lambda\in [0,1]^m$ such that $\lambda^\top A\in \Zz^n$. For $t \in \Zz_{\geq 2}$, the $t$-th CG closure of $P$ is recursively defined as $P^{(t)}=(P^{(t-1)})^{(1)}$. The CG closure is an important tool for solving integer programs, see again \cite{CCZ}.

\section{A $(1+\epsilon)$-oracle for pitch-$1$ and pitch-$2$ inequalities}
\label{sec:separation}

In this section, we show the following:

\begin{theorem}\label{thr:separation}
	Given a \mk \ instance \eqref{eq:mink}, there exists a $(1+\epsilon)$-oracle for the set ${\cal F}$ containing: all pitch-1 inequalities, all pitch-2 inequalities and all inequalities from the natural linear relaxation of \eqref{eq:mink}.
	\end{theorem}
We start with a characterization of inequalities of interest for Theorem \ref{thr:separation}.

\begin{lemma}\label{lem:pitch-2}
	Let $K$ be the set of feasible solutions of a \mk \ instance \eqref{eq:mink}. All pitch-$2$ inequalities valid for $K$ are implied by the set composed of:
	\begin{enumerate}[i)]
		\item
		Non-negativity constraints $x_i \geq 0$ for $i \in [n]$;
		\item
		All valid pitch-$1$ inequalities;
		\item \label{it:pitch-2-best}
		All inequalities of the form
		\begin{equation}\label{eq:pitch-2-best}\sum_{i \in I_1} x_i + 2 \sum_{i \in I_2} x_i \geq 2
		\end{equation}
		where $I\subseteq [n]$, $|I| \geq 2$, $\beta(I):=1-\sum_{i \in [n]\setminus I} p_i$, $I_1 :=\{ i \in I: p_i < \beta(I)\} \neq \emptyset$ and $I_2 := I\setminus I_1$.
	\end{enumerate}
	The inequalities in \ref{it:pitch-2-best}) are pitch-$2$ and valid.
\end{lemma}
Proofs of Lemma \ref{lem:IG-stays} and Theorem \ref{thr:separation} are given in Section \ref{sec:restricting-pitch-2} and Section \ref{sec:separation-algorithm}, respectively.

%

\subsection{Restricting the set of valid pitch-$2$ inequalities}
\label{sec:restricting-pitch-2}

We will build on two auxiliary statements in order to prove Lemma \ref{lem:pitch-2}.

\begin{myclaim}\label{cl:dominated}
	If $w^T x \geq \beta$ and $u^T x \geq \beta$ are distinct inequalities valid for and $u \geq w$, then the latter inequality is dominated by the former.
\end{myclaim}

\begin{proof}
	$u^T x \geq \beta$ can be obtained summing nonnegative multiples of $w^T x \geq \beta$ and $x_i \geq 0$ for $i \in [n]$, which are all valid inequalities.
	
	
\end{proof}

\begin{myclaim}
	\label{cl:support}
	Let \begin{equation}\label{eq:pitch-2-dominated}\sum_{i \in T_1} x_i + 2 \sum_{i \in T_2} x_i \geq 2
	\end{equation} be a valid inequality for \mk, with $T_1 \cap T_2 = \emptyset$ and $T_1, T_2 \subseteq [n]$. Then, \eqref{eq:pitch-2-dominated} is dominated by the inequality in \ref{it:pitch-2-best}) with $I=T_1 \cup T_2$.
\end{myclaim}

\begin{proof}
	One readily verifies that Inequality \eqref{eq:pitch-2-best} with $I$ as above is valid. Suppose now that $i \in T_1 \setminus I_1$. Then the integer solution that takes all elements in $([n]\setminus I)\cup\{i\}$ is feasible for \mk, but it does not satisfy \eqref{eq:pitch-2-dominated}, a contradiction. Hence $T_1\subseteq I_1$. Since $T_2=I\setminus T_i \supseteq I \setminus I_1=I_2$, \eqref{eq:pitch-2-best} dominates \eqref{eq:pitch-2-dominated} componentwise, and the thesis follows by Claim \ref{cl:dominated}.
	
\end{proof}

\textbf{\emph{Proof of Lemma \ref{lem:pitch-2}.}} The fact that an inequality of the form \eqref{eq:pitch-2-best} is pitch-$2$ and valid is immediate. Because of Claim \ref{cl:support}, it is enough to show the thesis with \eqref{eq:pitch-2-best} replaced by \eqref{eq:pitch-2-dominated}. Consider a pitch-$2$ inequality valid for $K$:
\begin{equation} \label{eq:w_pi2}
\sum_{i \in T} w_i x_i \geq 1,
\end{equation}
where $T \subseteq [n]$ is the support of the inequality, $w \in \R^{|T|}_+$. Without loss of generality one can assume that $T=[h]$ for some $h \leq n$ and $w_1 \leq w_2 \leq \dots \leq w_h$. Since (\ref{eq:w_pi2}) is pitch-$2$ we have that $w_1 + w_i \geq 1$ for all $i \in [h] \setminus \{ 1 \}$. We can also assume $w_h \leq 1$, since otherwise $\sum_{i \in [h-1]} w_i x_i + x_h \geq 1$ is valid and dominates \eqref{eq:w_pi2} by Claim \ref{cl:dominated}.

Let $j \in [h]$ be the maximum index such that $w_j < 1$. Note that such $j$ exists, since, if $w_1 \geq 1$, then \eqref{eq:w_pi2} is a pitch-1 inequality. If $1 - w_1 \leq 1/2$, then, by Claim \ref{cl:dominated}, \eqref{eq:w_pi2} is dominated by the valid pitch-$2$ inequality
\begin{equation} \label{eq:uw_pi2}
\sum_{i \in [j]} x_i + 2 \sum_{i=j+1}^h x_i \geq 2,
\end{equation}
which again is of the type \eqref{eq:pitch-2-dominated}. Hence $1- w_1 >  1/2$ and again via Claim \ref{cl:dominated}, (\ref{eq:w_pi2}) is dominated by
\begin{equation} \label{eq:nw_pi2}
w_1 x_1 + \sum_{i = 2}^{j} (1- w_1) x_i + \sum_{i=j+1}^h x_i \geq 1,
\end{equation} since $w_i + w_1 \geq 1$ for all $i \neq 1$, so one has $w_i \geq 1-w_1 > 1/2$.  Thus, we can assume that \eqref{eq:w_pi2} has the form \eqref{eq:nw_pi2}. Note that inequality \begin{equation}\label{eq:dajepitch}
\sum_{i=2}^h x_i \geq 1
\end{equation}
is a valid pitch-$1$ inequality, since we observed $w_1< 1$. Therefore, \eqref{eq:w_pi2} is implied by
\eqref{eq:uw_pi2} and \eqref{eq:dajepitch},  taken with the coefficients $w_1$ and $1- 2 w_1$ respectively. Recalling that \eqref{eq:uw_pi2} is a valid pitch-2 inequality of the form \eqref{eq:pitch-2-dominated} concludes the proof. \qedd




\subsection{A $(1+\epsilon)$-oracle}
\label{sec:separation-algorithm}

We will prove Theorem \ref{thr:separation} in a sequence of intermediate steps. Our argument extends the weak separation of KC inequalities in \cite{Carr2000}.

Let $\bar x$ be the point we want to separate. Note that it suffices to show how to separate over inequalities i)-ii)-iii) from Lemma \ref{lem:pitch-2}. Separating over i) is trivial. We first show how to separate over iii).

\begin{myclaim}
	\label{cl:pitch-2-violator}
	For $\alpha \in ]0,1]$, let $z^\alpha$ be the optimal solution to the following IP $P_\alpha$, and $v(z^\alpha)$ its value:
	\begin{equation} \label{eq:pitch-2-sep}
	\min \sum_{i \in [n]: \ p_i < \alpha} \bar x_i z_i + 2 \sum_{i \in [n]: \ p_i \geq \alpha} \bar x_i z_i \; \; \hbox{ s.t. }\; \; \sum_{i \in [n]} p_i (1 - z_i) \leq 1-\alpha, \; z \in \{0,1\}^n.
	\end{equation}
	If $v(z^\alpha)<2$, then $\bar x$ violates Inequality \eqref{eq:pitch-2-best} with $I:=\{i \in [n]: z_i^\alpha=1\}$, otherwise $\bar x$ does not violate any Inequality \eqref{eq:pitch-2-best} with $\beta(I)=\alpha$.
\end{myclaim}

\begin{proof}
	Fix a feasible solution $\bar z$ to \eqref{eq:pitch-2-sep}, and let $I:=\{i \in [n]: \bar z_i=1\}$. Then:
	$$\beta:=\beta(I)=1- \sum_{i \in [n]\setminus I} p_i  =1  - \sum_{i \in [n]} p_i (1 - \bar z_i) \geq \alpha.$$
	Hence:
	$$\begin{array}{lll}\sum_{i \in I: \ p_i < \beta} \bar x_i + 2 \sum_{i \in I: \ p_i \geq \beta} \bar x_i  & = & \sum_{i \in [n] : \ p_i < \beta} \bar x_i \bar z_i + 2 \sum_{i \in [n] : \ p_i \geq \beta} \bar x_i \\ & \leq & \sum_{i \in [n]: \ p_i < \alpha} \bar x_i \bar z_i + 2 \sum_{i \in [n]: \ p_i \geq \alpha} \bar x_i \bar z_i=v(\bar z),\end{array}$$
	where the central inequality holds at equality if $\alpha= \beta$. Hence, if $v(z^\alpha)<2$, the inequality with $I:=\{i \in [n]: z_i^\alpha=1\}$ from \eqref{eq:pitch-2-sep} is violated by $\bar x$. Else, all inequalities from \eqref{eq:pitch-2-sep} with $\beta(I)=\alpha$ are satisfied.
	
\end{proof}

Note that
$P_\alpha$ is a \mk \ instance, hence we can use the appropriate FPTAS to find, for each $\epsilon>0$, an $\epsilon$-approximate solution for it.

Since all data are rationals, we can assume there exists $q \in \N$ such that, for each $i \in [n]$, $p_i=r_i/q$ for some $r_i \in \N$.

\begin{myclaim}\label{cl:few-pitch-2}
	Let $r \in \{r_i + 1: \ i \in [n] \}$ and, for $\alpha=r/q$, let $\bar z^\alpha$ be the solution output by the FPTAS for problem $P_\alpha$ and $\bar v^\alpha$ its objective function value. If $\bar v^\alpha<2$ for some $\alpha$, then $\bar x$ violates Inequality \eqref{eq:pitch-2-best} with $I=\{ i \in [n]: \bar z^\alpha = 1\}$. Else, $(1+\epsilon) \bar x$ satisfies all inequalities in Lemma \ref{lem:pitch-2}.\ref{it:pitch-2-best}).  	
\end{myclaim}

\begin{proof}
	Let $r = r_i + 1$ for some $i \in [n]$. If $\bar v^\alpha < 2$ then by Claim \ref{cl:pitch-2-violator} $\bar x$ violates the corresponding inequality \eqref{eq:pitch-2-best}. Otherwise, $v^\alpha \geq 2/(1+\epsilon)$, and $(1+\epsilon)\bar{x}$ is feasible for any pitch-$2$ Inequality \eqref{eq:pitch-2-best} induced by $I$ with $\beta(I)=\alpha$.
	
	Now let $I^* \subseteq [n]$ with $\beta(I^*)=\frac{r^*}{q} < 1$. There exists $i^* \in [n]$ such that $r_{i^*} < r^* \leq r_{i^*+1} \leq q$ (with $r_{n+1} = q$). Let $\alpha:=\frac{r_{i^*} + 1}{q}\leq \alpha^*:=\beta(I^*)$. The set of feasible solutions of $P_\alpha$ contains that of $P_{\alpha^*}$, and $\{i \in [n] : p_i < \alpha\}= \{i \in [n] : p_i < \alpha^*\}$. Hence, $v^{\alpha^*}\geq v^{\alpha}$ and consequently $v^{\alpha^*}<2$ implies $v^{\alpha}<2$.
	Thus, for separating all inequalities in Lemma \ref{lem:pitch-2}.\ref{it:pitch-2-best}), it suffices to check \eqref{eq:pitch-2-sep} for all $\alpha=\frac{r}{q}$ as in the statement of the claim.
	
	
\end{proof}

The following claim follows in a similar fashion to the previous one by observing that, for $\beta(I^*)=\frac{1}{q}$, \eqref{eq:pitch-2-sep} separates over undominated pitch-$1$ inequalities.

\begin{myclaim}\label{cl:few-pitch-1}
	Let $\alpha=1/q$, and $\bar z^\alpha$ be the solution output by the FPTAS for problem $P_\alpha$, and $\bar v^\alpha$ its objective function value. If $\bar v^\alpha<2$, then $\bar x$ violates the pitch-$1$ inequality with support $I=\{ i \in [n]: \bar z^\alpha = 1\}$. Else, $(1+\epsilon) \bar x$ satisfies all valid pitch-$1$ inequalities.  \end{myclaim}

Next claim shows how to round a point in the unit cube that almost satisfies all pitch-$1$ and pitch-$2$ inequalities, to one that satisfies them and is still contained in the unit cube.

\begin{myclaim}\label{cl:round}
	Let $\bar x \in [0,1]^n$ be such that $(1+\epsilon)\bar x$ satisfies all inequalities from Lemma \ref{lem:pitch-2}, and define $\bar y \in \R^n$ as follows: $\bar y_i = \min \{1, \frac{1+\epsilon}{1-\epsilon}\bar x_i\}$ for $i \in [n]$. Then $\bar y \in [0,1]^n$ and $\bar y$ satisfies all inequalities from Lemma \ref{lem:pitch-2}.
\end{myclaim}

\begin{proof}
	Clearly $\bar y \in [0,1]^n$. Let $J=\{ i \in [n]: \bar y_i =1 \}$. If $J=\emptyset$, $(1+\epsilon)\bar x< \bar y \leq \mathbf{1}$, hence $\bar y$ satisfies all pitch-$2$ inequalities. Thus, $J \neq \emptyset$. Consider a pitch-$2$ inequality of the form \eqref{eq:pitch-2-best}, and note that the left-hand side of the inequality computed in $\bar y$ is lower bounded by $\sum_{i \in J}\alpha_i$, where $\alpha_i$ is the coefficient of $x_i$. First assume there exists $j \in J \cap I_2$. Then $\sum_{i \in J}\alpha_i \geq \alpha_j = 2$. Similarly, if $j,j' \in J$, then $\sum_{i \in J}\alpha_i \geq \alpha_j + \alpha_{j'} \geq 2$.
	In both cases, $\bar y$ satisfies the pitch-$2$ inequality. Hence, we can assume $J =\{j\}\subseteq I_1$. Then:
	$$\sum_{i \in I} \alpha_i \bar x_i  \geq \frac{2}{1+\epsilon}\hbox{, from which we deduce }$$ $$\sum_{i \in I\setminus\{j\}} \alpha_i \bar x_i \geq \frac{2}{1+\epsilon}-\bar x_j \geq \frac{2}{1+\epsilon}- 1 = \frac{1-\epsilon}{1+\epsilon} \hbox{ and }$$
	$$\sum_{i \in I} \alpha_i \bar y_i  =  \sum_{i \in I\setminus\{j\}} \alpha_i \bar y_i + 1 = \frac{1+\epsilon}{1-\epsilon}\sum_{i \in I\setminus\{j\}} \alpha_i \bar x_i + 1 \geq 2,$$
	as required. A similar (simpler) argument shows that $\bar y$ also satisfies all pitch-$1$ inequalities $\sum_{i \in I} x_i \geq 1$.
	
	
\end{proof}

\textbf{\emph{Proof of Theorem \ref{thr:separation}.}} We can now sum up our $(1+\epsilon)$-oracle, see Algorithm \ref{algo:sep}. Correctness and polynomiality follow from the discussion above. \qedd

\begin{algorithm}[h!]
	\caption{}	\label{algo:sep}
	\begin{algorithmic} [1]
		\STATE {Let $\epsilon'=\frac{\epsilon}{2+\epsilon}$.}
		\STATE {For $r \in \{r_i + 1: \ i \in [n] \}$ and for $\alpha=r/q$, run the FPTAS for $P_\alpha$ with approximation factor $\epsilon'$. If any of the output solution $\bar z^\alpha$ has value $\bar v^\alpha <2$, output inequality \eqref{eq:pitch-2-best} with $I=\{i \in [n]: \bar z^\alpha=1\}$ and stop. }
		\STATE {For $\alpha=1/q$, run the FPTAS for $P_\alpha$ with approximation factor $\epsilon'$. If the output solution $\bar z^\alpha$ has value $\bar v^\alpha <2$, output inequality $\sum_{i : \bar z^\alpha=1} x_i \geq 1$ and stop.}
		\STATE{Output point $\bar y$ constructed as in Claim \ref{cl:round} with $\epsilon'$ and stop. Note that $\bar x \leq \bar y \leq \frac{1+\epsilon'}{1-\epsilon'}\bar x = (1+\epsilon)\bar x$.}
	\end{algorithmic}
\end{algorithm}

\subsection{Separating inequalities of pitch $k \geq 3$ with fixed support}
\label{sec:separating-fixed-support}


Here, we give an example showing that inequalities of pitch-$3$ and higher do not have the nice structure of pitch-2. Let
\begin{equation}
\label{eq:pitch-3-wild}
P=\{x \in [0,1]^7: 5 x_1 + 6 x_2 + 11 x_3 + 16 x_4 + 17 x_5 + 18 x_6 + 21 x_7 \geq 41 \}.
\end{equation}
Inequality $x_1 + x_3 + x_4 + 2 x_5 + x_6 + 2 x_7 \geq 3$ is a facet of the first CG closure $P^{(1)}$ (although not of the integer hull of $P$) and thus a valid pitch-$3$. Observe that the coefficient of $x_5$ is higher than $x_6$ in this pitch-$3$, while it is the opposite in \eqref{eq:pitch-3-wild}. Such situations we call \emph{inversions} and they do not occur in (relevant) pitch-2 inequalities. Inequality $x_1 + x_2 + 2 x_3 + 3 x_4 + 4 x_5 + 3 x_6 + 4 x_7 \geq 8$ is an inverted facet of both the integer hull and the first CG closure.

For later use (in Section \ref{sec:KC+bounded}), we observe here that when $I \subseteq [n]$ is fixed, we can \emph{efficiently} and \emph{exactly} solve the separation problem over inequalities with support $I$ just by solving an LP. Clearly, we are only interested in valid inequalities $\alpha^T x \geq 1$ with $\alpha \geq 0$ and points $0\leq x^*\leq 1$.
Let $\beta=1-p([n]\setminus I)$. We can assume $\beta>0$, otherwise there is no valid inequality as above with support $I$.  Call $J \subseteq I$  \emph{massive} if $\sum_{i \in J}p_i \geq \beta$. Consider the following LP:
\begin{equation}
\label{eq:fixed-support}
\begin{array}{lrlll}
\min \; & \sum_{i \in I} \alpha_i x^*_i \\
\hbox{s.t.}\\
& \sum_{i \in J}\alpha_i & \geq & 1 & \hbox{ for all massive $J\subseteq I$} \\
& \alpha & \geq & 0
\end{array}
\end{equation}
Note that, for each feasible solution $\bar \alpha$ to the previous LP, we have that $\bar \alpha^T x \geq 1$ is a valid inequality for the original \mk \ instance, and conversely that all inequalities with support $I$ can be obtained in this way. Hence, let $\alpha^*$ be the optimal solution to the previous LP. If $(\alpha^*)^T x^* < 1$, we obtain an inequality  whose support is contained in $I$, that is violated by $x^*$. The support of the inequality can be extended to $I$ by setting $\alpha_i=\epsilon$ for all $i \in I$ with $\alpha_i=0$.
On the other hand, if $(\alpha^*)^T x^* \geq 1$, $x^*$ satisfies all inequalities with support $I$.

\section{Integrality gap for \mk \ with bounded pitch inequalities}

\subsection{When $p=c$}\label{sec:p-c}

\begin{theorem}
	\label{thm:p-c}
	Consider an instance of \mk \ \eqref{eq:mink} with $p=c$. Denote by $K$ the linear relaxation of \eqref{eq:mink} to which all pitch-1 and pitch-2 inequalities have been added. The integrality gap of $K$ is at most $3/2$.
\end{theorem}
\textbf{\emph{Proof.}}
Let $p=c$, and let $\bar x$ be the optimal integer solution to \eqref{eq:mink}. We can assume $p^T \bar x > 1.5$, else we are done.

\begin{myclaim}
	The support of $\bar x$ has size $2$.
\end{myclaim}
\begin{proof}
	Let $k$ be the size of the support of $\bar x$. If $k=1$, then $\bar x$ is also the optimal fractional solution. Now assume $k \geq 3$. Remove from $\bar x$ the cheapest item as to obtain $\bar x'$. We have
	$$p^T x' \geq \left( 1-\frac{1}{k} \right )p^T \bar x > \frac{2}{3} \cdot 1.5 = 1,$$
	contradicting the fact that $\bar x$ is the optimal integral solution.
	
	
\end{proof}

\smallskip

Hence, we can assume that the support of $\bar x$ is given by $\{i,j\}$, with $0<p_i\leq p_j\leq 1$. Since $p_i + p_j >1.5$, we deduce $p_j > .75$. Since $p_j\leq 1$, we deduce $p_i>.5$.

\begin{myclaim}
	Let $\ell < j$ and $\ell \neq i$. Then $p_\ell <.25$.
\end{myclaim}
\begin{proof}
	Recall that for $S \subseteq [n]$ we denote its characteristic vector with $\chi^{S}$. If $0.25 \leq p_\ell < p_i$, then $\chi^{\{\ell,j\}}$ is a feasible integral solution of cost strictly less than $\bar x$. Else if $0.5<p_i\leq p_\ell<p_j$, then $\chi^{\{\ell,i\}}$ is a feasible integral solution of cost strictly less than $\bar x$. In both cases we obtain a contradiction.
	
\end{proof}

Because of the previous claim, we can assume w.l.o.g. $j=i+1$.

\begin{myclaim}
	$p_n+\sum_{\ell=1}^{i-1}p_\ell<1$.
\end{myclaim}

\begin{proof}
	Suppose $p_n+\sum_{\ell=1}^{i-1}p_\ell\geq 1$. Since $p_\ell<.25$ for all $\ell=1,\dots,i-1$, there exists $k\leq i-1$ such that $x_n+\sum_{\ell=1}^{k}p_k \in [1,1.25[$. Hence $x^{\{1,\dots,k,n\}}$ is a feasible integer solution of cost at most $1.25$, a contradiction.	
	
\end{proof}

Because of the previous claim, the pitch-$2$ inequality $\sum_{k=i}^n x_k \geq 2$ is valid. The fractional solution of minimum cost that satisfies this inequality is the one that sets $x_i=x_{j}=1$ (since $j=i+1$) and all other variables to $0$. This is exactly $\bar x$.\qedd

\subsection{CG closures of bounded rank of the natural \mk \ relaxation}\label{sec:CG}

For $t \in \N$, let $K^t$ be the linear relaxation of the \eqref{eq:mink} given by: the original knapsack inequality; non-negativity constraints; all pitch-$k$ inequalities, for $k\leq t$.

\begin{lemma}\label{lem:IG-stays}
	For $t \geq 2$, the integrality gap of $K^t$ is at least $\max\{\frac{1}{2},\frac{t-2}{t-1}\}$ times the integrality gap of $K^{t-1}$. 	 \end{lemma}
\textbf{\emph{Proof.}}	
Fix $t \geq 2$, and let $C$ be the cost of the optimal integral solution to \eqref{eq:mink}. Let $C/v'$ be the integrality gap of $K^t$. Since $v'$ is the optimal value of $K^t$, by the strong duality theorem (and Caratheodory's theorem), there exist nonnegative multipliers $\alpha, \alpha_1,\dots, \alpha_n,\gamma_1,\dots, \gamma_{n+1}$ such that the inequality $c^T x \geq v'$ can be obtained as a conic combination of the original knapsack inequality (with multiplier $\alpha$), non-negativity constraints (with multipliers $\alpha{_1,\dots,\alpha_n}$), and at most $n+1$ inequalities of pitch at most $t$ (with multipliers $\gamma_1,\dots,\gamma_{n+1}$). By scaling, we can assume that the rhs of the latter inequalities is $1$. Hence $v'=\alpha + \sum_{i=1}^r \gamma_i$. 

\begin{myclaim}\label{cl:a-little-more-pitch}	Let $d^Tx \geq 1$ be a valid pitch-$t$ inequality for \eqref{eq:mink}, and assume w.l.o.g. that $d_1\leq d_2 \leq \dots \leq d_n$. Then inequality $\sum_{i=2}^{n} d_i x_i \geq  \max\{\frac{1}{2},\frac{t-2}{t-1}\}$ is a valid inequality of pitch at most $t-1$ for \eqref{eq:mink}.
\end{myclaim}

\begin{proof}

	The inequality $\sum_{i=2}^n d_i x_i \geq 1 - d_1$ is a valid inequality for \eqref{eq:mink}, and by construction it is of pitch at most $t-1$. If $t\geq 3$, we obtain $\sum_{i=1}^{t-1} d_i < 1$ and consequently $d_1 < 1 / (t-1)$, from which we deduce
	$$ 1 - d_1 > 1 - \frac{1}{t-1} =  \frac{t-2}{t-1}\geq \frac{1}{2}.$$
	If conversely $t=2$, by Lemma \ref{lem:pitch-2} we can assume w.l.o.g. that $d_1=1/2$, and we can conclude $1 -d_1=\frac{1}{2}> \frac{t-2}{t-1}$.
\end{proof}

Now consider the conic combination with multipliers $\alpha,\alpha_1,\dots,\alpha_n, \gamma_1,\dots,\gamma_{n+1}$ given above, where each inequality of pitch-$t$ is replaced with the inequality of pitch at most $t-1$ obtained using Claim \ref{cl:a-little-more-pitch}. We obtain an inequality $(c')^Tx \geq v''$, where one immediately checks that $c'\leq c$ and
$$\begin{array}{lll}v'' & \geq & \alpha + \sum_{i=1}^{n+1} \gamma_i  \max\left\{\frac{1}{2},\frac{t-2}{t-1}\right\} \geq   \max\left\{\frac{1}{2},\frac{t-2}{t-1}\right\} \left(\alpha + \sum_{i=1}^{n+1} \gamma_i\right) \\ & = &  \max\left\{\frac{1}{2},\frac{t-2}{t-1}\right\} v'.\end{array}$$
Hence the integrality gap of $K^t$ is
$$\frac{C}{v'}\geq \frac{C}{v''}\max\left\{\frac{1}{2},\frac{t-2}{t-1}\right\}$$
and the thesis follows since the integrality gap of $K^{t-1}$ is at most $C/v''$.

\qedd

\begin{lemma}\label{lem:knapsack}
	
	For a fixed $\epsilon>0$ and square integers $n\geq 4$, consider the \mk\, instance $K$ defined as follows:
	$$\begin{array}{lrrrrrrll}
	\min & & \epsilon y  & + & \sqrt{n} z & + & \sum_{i=1}^n x_i\\
	& \hbox{st} \\
	&& (n-\sqrt{n}) y & + & \frac{n}{2} z & + & \sum_{i=1}^n x_i & \geq n  \\ \\[-.1ex]
	&& y,& & z, & & x & \in \{0,1\}.\end{array}$$
	For every fixed $t \in \N$, the integrality gap of $K^t$ is $\Omega(\sqrt{n})$.
\end{lemma}
\textbf{\emph{Proof.}}
Because of Lemma \ref{lem:IG-stays}, it is enough to show that the integrality gap of $K^1$ is $\Omega(\sqrt{n})$. Clearly, the value of the integral optimal solution of the instance is $\sqrt{n}+\epsilon$. We claim that the fractional solution
$$(\bar y,\bar x,\bar z)=\left(1,\underbrace{\frac{1}{n-\sqrt{n}+1},\dots,\frac{1}{n-\sqrt{n}+1}}_{n \hbox{ times}},\frac{2}{\sqrt{n}}\right)$$
is a feasible point of $K^1$. Since $\epsilon \bar y  + \sqrt{n} \bar z + \sum_{i=1}^n \bar x_i =\epsilon + 2 + \frac{n}{n-\sqrt{n}+1}$, the thesis follows.

Observe that $(n-\sqrt{n}) \bar y +  \frac{n}{2} \bar z +\sum_{i=1}^n \bar x_i=(n-\sqrt{n}) + \frac{n}{2} \frac{2}{\sqrt{n}} + \frac{n}{n-\sqrt{n}+1}>n$, hence $(\bar y,\bar x,\bar z)$ satisfies the original knapsack inequality.

Now consider a valid pitch-1 inequality whose support contains $y$. Since $\bar y=1$, $(\bar y,\bar x,\bar z)$ satisfies this inequality. Hence, the only pitch-1 inequalities of interest do not have $y$ in the support. Note that such inequalities must have $z$ in the support, and some of the $x_i$. Hence, all those inequalities are dominated by the valid pitch-$1$ inequalities
$$z + \sum_{i \in I}x_i \geq 1 \; \forall I \subseteq [n], |I|=n - \sqrt{n} +1,$$ which are clearly satisfied by $(\bar y,\bar x,\bar z)$.

\qedd

\begin{theorem}
	For a fixed $q \in \N$, let CG$^q(K)$ be the $q-$th CG closure of the \mk\, instance $K$ as defined in Lemma \ref{lem:knapsack}. The integrality gap of CG$^q(K)$ is $\Omega(\sqrt{n})$.
\end{theorem}
\textbf{\emph{Proof.}}
We will use the following fact, proved (for a generic covering problem) in \cite{Bienstock2006}. Let $t,q \in \N$ and suppose $(\bar y,\bar z, \bar x) \in  K^t$. Define point
$(y',x',z')$, where each component is the minimum between $1$ and $(\frac{t+1}{t})^q$ times the corresponding component of $(\bar y,\bar z, \bar x)$. Then $(y',x',z') \in CG^q(K)$. Now fix $t,q$. We have therefore that
$$\begin{array}{lll}\epsilon y'   + \sqrt{n} z' + \sum_{i=1}^n x'_i & \leq &  \left(\frac{t+1}{t}\right)^q \left(\epsilon \bar y   + \sqrt{n} \bar z+  \sum_{i=1}^n \bar x_i\right) \\ &
= & \left(\frac{t+1}{t}\right)^q \left(\epsilon + 2 + \frac{n}{n-\sqrt{n}+1}\right)\end{array}$$
and the claim follows in a similar fashion to the proof of Lemma \ref{lem:knapsack}.\qedd

\subsection{When all knapsack cover inequalities are added}
\label{sec:kc-gap-2}

In this section we consider the min-knapsack formulation with knapsack cover inequalities (we use KC to denote this LP formulation).
In \cite{Carr2000} it is shown that the integrality gap of KC is 2. In the following we provide a simpler proof.

Let $\bar x$ be a feasible fractional solution for KC of cost $C(\bar x)$. Starting from $\bar x$, we show a simple and fast rounding procedure to obtain a feasible integral solution of cost at most $2C(\bar x)$.
\paragraph{The rounding procedure:}
Let $S=\{i\in [n]:\bar x_i\geq 1/2\}$. Set $x_i=1$ for any $i\in S$. Consider the residual variables $\bs:=[n]\setminus S$.
The problem is to assign integral values to the the residual variables. We call this problem the \emph{residual problem} (RP). By abusing notation, from now on, let $\bar x$ denote the fractional solution for KC restricted to residual variables.
Consider the following \emph{residual relaxation} (RR):
\begin{align}
\min &\sum_{i\in \bs} C_i x_i\\
s.t. &\sum_{i\in \bs}p_i' x_i \geq b'\\
& 0\leq x_i\leq 1/2
\end{align}
where $p'_i=\min\{p_i, b-p(S)\}$ and $b'=b-p(S)$. Note that $\bar x$ satisfies (RR). So if $x^*$ is the optimal solution of (RR) than it follows that
$$ \sum_{i\in \bs} C_i x_i^* \leq \sum_{i\in \bs} C_i \bar x_i$$
Therefore, if it exists an integral solution $x^{int}$ to (RP) of cost $C(x^{int})\leq 2C(x^*)$ then $C(x^{int})\leq 2C(\bar x)$ and we are done.
We can rewrite the \emph{residual relaxation} (RR) in the following equivalent way:
\begin{align}
\min &\sum_{i\in \bs} \frac{C_i}{2} y_i\\
s.t. &\sum_{i\in \bs}\frac{p_i'}{2} y_i \geq b'\\
& 0\leq y_i\leq 1
\end{align}
Clearly the optimal fractional solution to (RR) can be obtained by ordering the variables according to their densities. W.l.o.g., assume that $\frac{C_1}{p_1'}\leq \frac{C_2}{p_2'}\leq \ldots \leq \frac{C_{|\bs|}}{p_{|\bs|}'} $ and let $t+1\in[|\bs|]$ be the smallest integer such that
$\sum_{i= 1}^{t+1}\frac{p_i'}{2} \geq b'$ and therefore (recall $p'_i\leq b'$):
\begin{align}\label{eq:prof}
\sum_{i= 1}^{t}p_i' \geq 2b'-p'_{t+1}\geq b'
\end{align}
Note that the optimal fractional solution to (RR) picks the first $t$ variables integrally and the last $t+1$ potentially fractional (but it could be integral). It follows that:
\begin{align}\label{eq:cost}
\sum_{i=1}^t \frac{C_i}{2} \leq \sum_{i\in \bs} C_i x_i^*
\end{align}

It follows that the integral solution $x^{int}$ obtained by setting $x_1=1$ for $i\in [t]$ and zero otherwise is feasible by $\eqref{eq:prof}$ and of cost at most twice the optimal fractional of (RR) by \eqref{eq:cost}.

\subsection{When all bounded pitch and knapsack cover inequalities are added}\label{sec:KC+bounded}

Consider the following \mk\ instance with $\epsilon_n = \frac{1}{\sqrt{n}}$:

\begin{align}
\label{eq:ola}
\begin{split}
\min \quad & \sum_{i \in [n]} x_i + \frac{1}{\sqrt{n}} \sum_{j \in [n]} z_j  \\
\hbox{s.t.} \quad & \sum_{i \in [n]} x_i + \frac{1}{n} \sum_{j \in [n]} z_j \geq 1 + \epsilon_n \\
\ & x, z \in \{0,1\}^n.
\end{split}
\end{align}

\begin{lemma}
	For any fixed $k \in \N$ and $n \in \N$ sufficiently large, point $(\bar x ,\bar z) \in \R^{2n}$ with $\bar x_i = \frac{1+\epsilon_n}{n}, \bar z_i = \frac{k}{n}$ satisfies the natural linear relaxation, all KC and all inequalities of pitch at most $k$ valid for \eqref{eq:ola}. Observing that the optimal integral solution is $2$, this gives an IG of $\frac{2}{1 + \frac{k}{n}} \approx 2$.
\end{lemma}
\textbf{\emph{Proof.}}
We prove the statement by induction. Fix $k \in \N$. Note that $(\bar{x},\bar{z})$ dominates componentwise the point generated at step $k-1$, and the latter by induction hypothesis satisfies all inequalities of pitch at most $k-1$. Let
\begin{equation}
\label{eq:ola-form}
\sum_{i \in I} w_i x_i + \sum_{j \in J} w_j z_j \geq \beta
\end{equation}
be a valid KC or pitch-$k$ inequality with support $I \cup J$, which gives that $w_i, w_j \in \R_{>0}, \ \forall i \in I, \ \forall j \in J$ and $\beta > 0$. Observe the following.

\begin{myclaim}
	\label{cl:size-of-I}
	$|I| \geq n-1$. In addition, $|I| = n-1$ or $|J| \leq n (1-\epsilon_n)$ implies $w_i \geq \beta, \ \forall i \in I$.
\end{myclaim}

\begin{proof}
	Since all coefficients in \eqref{eq:ola-form} are strictly positive and $\beta > 0$, $|I| \leq n-2$ gives that the feasible solution $(\chi^{[n] \setminus I}, \vec{0})$ for \eqref{eq:ola} is cut off by \eqref{eq:ola-form}, a contradiction.
	
	Furthermore, if $|I| = n-1$ and $w_{i^*} < \beta$ for some $i^* \in I$, then $(\chi^{([n] \setminus I) \cup \{i^* \} }, \vec{0})$ is cut off, again a contradiction.
	Finally, if $|I| = n$, $|J| \leq n (1-\epsilon_n)$ and $w_{i^*} < \beta$ for some $i^* \in I$, then $(\chi^{\{i^* \}}, \chi^{\{[n] \setminus J \}})$ does not satisfy \eqref{eq:ola-form}, but it is feasible in \eqref{eq:ola}.
	%
\end{proof}

We first show the statement for \eqref{eq:ola-form} being a KC. By the definition of KC:
$
\beta = 1 + \epsilon_n - |[n] \setminus I|-\frac{|[n] \setminus {J}|}{n}
$, $w_i=\min\{1, \beta\}, \ \forall i \in I$ and $w_j=\min\{\frac{1}{n}, \beta\}, \ \forall j \in J$. If $I=[n]$, then $\sum_{i \in I} w_i \bar x_i = \min\{1, \beta\} \cdot (1+\epsilon_n) \geq \beta$ since $\beta \leq 1 + \epsilon_n$. Otherwise, $|[n] \setminus I|=1$ so $w_i=\beta, \ \forall i \in I$ and $\sum_{i \in I} w_i \bar x_i= \beta \frac{(n-1)(1+\epsilon_n)}{n} > \beta$ for sufficiently large $n$.

Conversely, let \eqref{eq:ola-form} be a valid pitch-$k$ inequality. By Claim \ref{cl:size-of-I}, if $|I| = n-1$ or $|J| \leq n (1-\epsilon_n)$ then $w_i \geq\beta, \ \forall i \in I$ so the proof is analogous to the one for KC. Otherwise, $|I|=n$ and $|J| > n (1-\epsilon_n)$.
Consider the LP \eqref{eq:fixed-support} in Section \ref{sec:separation} specialized for our case -- that is, we want to detect if $(\bar{x},\bar{z})$ can be separated via an inequality with support $I \cup J$.
Since $(\chi^{i,\bar \imath},\mathbf{0})$ is feasible in \eqref{eq:ola} for $i, \bar \imath \in I$, then \begin{equation}\label{eq:alphai}\alpha_i + \alpha_{\bar  \imath} \geq 1.\end{equation} Furthermore, for $n$ large enough one has $|J| > n (1-\epsilon_n) \geq k$ so
\begin{equation}\label{eq:alphaj}\sum_{j \in K} \alpha_j \geq 1\end{equation} for any $k$-subset $K$ of $J$. We claim that the minimum in \eqref{eq:fixed-support} is attained at $\bar \alpha_i=1/2, \ \forall i \in I$ and $\bar \alpha_j=1/k, \ \forall j \in J$. Indeed, the objective function of \eqref{eq:fixed-support} computed in $\bar \alpha$ is given by
$$ |I| \cdot \frac{1}{2} \cdot \frac{1+\epsilon}{n} + |J| \cdot \frac{1}{k} \cdot \frac{k}{n} = \frac{1+\epsilon_n}{2} + \frac{|J|}{n}.$$
On the other hand, by summing \eqref{eq:alphai} for all possible pairs with multipliers $\frac{1+\epsilon_n}{2(n-1)}$ and \eqref{eq:alphaj} for all subsets of $J$ of size $k$ with multipliers ${|J|-1 \choose {k-1}}^{-1} \cdot \frac{k}{n}$, simple linear algebra calculations lead to
$$\sum_{i \in I}\bar x_i \alpha_i + \sum_{j \in J}\bar z_j \alpha_j \geq \frac{1+\epsilon_n}{2} + \frac{|J|}{n},$$
showing the optimality of $\bar \alpha$. Recalling $|J|>n(1-\epsilon_n)$, we conclude that
$$\sum_{i \in I}\bar x_i \bar \alpha_i + \sum_{j \in J}\bar z_j \bar \alpha_j = \frac{1+\epsilon_n}{2} + \frac{|J|}{n} > \frac{1+\epsilon_n}{2} + 1 -\epsilon_n>1,$$
hence $(\bar{x},\bar{z})$ satisfies all inequalities with support $I \cup J$. \qedd

\bigskip

{\bf Acknowledgments.} Some of the work was done when the second and the third author visited the IEOR department of Columbia University, partially funded by a gift of the  Swiss National Science Foundation.

\bibliographystyle{plain}
\bibliography{citation}



\end{document}